\documentclass[a4paper]{amsproc}
\usepackage{amssymb}
\usepackage{amscd}
\usepackage[dvips]{graphicx}

\theoremstyle{plain}
 \newtheorem{thm}{Theorem}[section]
 
 \newtheorem{lem}{Lemma}[section]
 \newtheorem{cor}{Corollary}[section]
\theoremstyle{definition}
 \newtheorem{exm}{Example}[section]
 \newtheorem{rem}{Remark}[section]

\numberwithin{equation}{section}

\newcommand{\R}{\mathbb{R}}

\title[SYMMETRIES OF LINE BUNDLES AND NOETHER THEOREM]{SYMMETRIES OF LINE BUNDLES AND NOETHER
THEOREM FOR TIME-DEPENDENT NONHOLONOMIC SYSTEMS}
\author[Jovanovi\'c]{\bfseries Bo\v zidar Jovanovi\'c}

\address{
Mathematical Institute SANU \\
Serbian Academy of Sciences and Arts \\
Kneza Mihaila 36, 11000 Belgrade\\
Serbia}
\email{bozaj@mi.sanu.ac.rs}

\begin{document}

\begin{abstract}
We consider Noether symmetries of the equations defined by the
sections of characteristic line bundles of nondegenerate
1-forms and of the associated perturbed systems.
It appears that this framework can be used for time-dependent systems with constraints and nonconservative forces, allowing a quite simple and transparent formulation of the momentum equation and the Noether theorem in their general forms.
\end{abstract}

\maketitle

\section{Introduction}

The Noether theorem on integrals related to invariant variational problems is one of the basic theorems in mechanics, both for finite and infinite-dimensional systems \cite{KS, Noether}. There has been a
lot of efforts on its generalization and our reference list \cite{Ar}--\cite{SS} covers
just a part of contributions for finite-dimensional systems.

Recently, we presented the problem for non-constrained systems through the perspective of contact geometry \cite{Jo1}.
We considered Noether symmetries of the equation
\begin{equation}\label{clb}
\dot x=Z,
\end{equation}
where $Z$ is a section of the characteristic line bundle of a nondegenerate
1-form $\mathcal L=\ker d\alpha$, as symmetries that preserve the action functional $A_\alpha[\gamma]=\int_\gamma \alpha$. In the case of time-dependent Hamiltonian systems, Noether symmetries are transformations that preserve Poincar\'e--Cartan (modulo addition of a closed 1-form) and, via Legendre transformation, this is similar to the notion of symmetry of Lagrangian systems given by Crampin in \cite{Cr}.

It appears that the above framework can be used for time-dependent systems with constraints and nonconservative forces, allowing a quite simple and transparent formulation of the momentum equation and of the Noether theorem in their general forms.

We briefly recall on the notion of Noether symmetries for systems defined by the sections of non-degenerate 1-forms studied in \cite{Jo1} and formulate the statements for perturbed systems
\begin{equation}\label{ZR}
\dot x=Z+P,
\end{equation}
where $P$ is not a section of $\mathcal L=\ker d\alpha$ (Theorems \ref{druga}, \ref{treca}, Section \ref{S2}). In Section \ref{S3} we apply Theorem \ref{druga} and obtain the main results: a general momentum equation and Noether theorem for time-dependent nonholonomic systems subjected to nonconservative forces (Theorems \ref{cetvrta}, \ref{cetvrta*}). In particular, when we deal with symmetries that are prolongation of time-dependent vector fields on the configuration space, we get a time-dependent variant of the so called gauge symmetries studied in \cite{BS, BGMD,  FRS, FS} (see Corollary \ref{peta}) and the moving energy integral given in \cite{FS, FS2} (see Corollary \ref{sesta}).

\section{Noether symmetries of characteristic line bundles}\label{S2}

\subsection{Dynamical systems defined by characteristic line bundles}

Let $(M,\alpha)$ be a $(2n+1)$--dimensional manifold
endowed with a 1-form $\alpha$, such that $d\alpha$ has the maximal rank $2n$. The kernel of
$d\alpha$ defines a one dimensional distribution
\[
\mathcal L=\cup_x \mathcal L_x, \qquad \mathcal L_x=\ker d\alpha\vert_x
\]
 of the
tangent bundle $TM$ called \emph{characteristic line bundle}.

Also,
at every point $x\in M$ we have the \emph{horizontal space}
$
\mathcal H_x=\ker\alpha\vert_x.
$
In the case when $\alpha$ differs from zero on $M$, then the
collection of horizontal subspaces $ \mathcal H=\cup_x \mathcal
H_x=\cup_x\ker\alpha\vert_x $ is a nonintegrable $2n$--dimensional distribution
of $TM$, called \emph{horizontal distribution}. If, in addition,
$\alpha\wedge(d\alpha)^n\ne 0$, then $\alpha$ is a \emph{contact form}, $(M,\alpha)$ is a
\emph{strictly contact manifold}, and  $\mathcal H$
is a \emph{contact} distribution \cite{LM}.

The integral curves $\gamma: [a,b]\to M$ of the
characteristic line bundle $\mathcal L$ are extremals of the
action functional
\begin{equation}\label{action}
A_\alpha[\gamma]=\int_\gamma \alpha=\int_a^b \alpha(\dot\gamma)dt
\end{equation}
in a class of variations $\gamma_s$ with fixed endpoints.
Recall that a {\it variation} of a curve $\gamma: [a,b]\to M$ is a
family of curves $\gamma_s(t)=\Gamma(t,s)$, where $\Gamma: [a,b]\times [0,\epsilon] \to M$ is a mapping, such that
$\gamma(t)=\Gamma(t,0)$, $t\in [a,b]$. The endpoints are fixed if $\gamma_s(a)\equiv \gamma(a)$, $\gamma_s(b)\equiv \gamma(b)$.\footnote{Here the usual assumption that the endpoints are fixed can be relaxed: we
can consider also the variations $\gamma_s(t)$, such that $\delta\gamma(a)$ and $\delta\gamma(b)$ are horizontal vectors, where
$\delta\gamma(t)$ denotes the vector field $\frac{\partial \Gamma}{\partial s}\vert_{s=0}\in
T_{\gamma(t)} M$, e.g., see \cite{Jo1}.}

Consider the equation \eqref{clb},
where $Z$ is a section of $\mathcal L$. We say that a
vector field $\zeta$
 is a \emph{Noether symmetry} of
equation \eqref{clb} if
 the induced one-parameter group of
diffeomeomorphisms $g_s^\zeta$ preserves the 1-form $\alpha$.
Then, by analogy with the classical formulation
\cite{SC1, Noether}, $g_s^\zeta$ preserves the action
functional \eqref{action}.

Note that $\mathcal L$ is determined by the
cohomology class $[\alpha]$ ($\mathcal L=\ker d\alpha'$, where $\alpha'=\alpha+\beta$, $\beta$ is a closed 1-form on $M$),
while $\mathcal H$ depends on $\alpha$. Thus, the integral curves of $\mathcal L$ are also extremals the action \eqref{action} with $\alpha$ replaced by $\alpha'\in[\alpha]$. We say that
$\zeta$ is a \emph{weak Noether symmetry} of
equation \eqref{clb} if we have the invariance of the perturbation $\alpha'=\alpha+\beta$, $d\beta=0$, modulo the differential of a function $f$:
\begin{equation}\label{wns1}
L_\zeta(\alpha')=L_{\zeta}(\alpha+\beta)=df.
\end{equation}
That is, $g^\zeta_s$ preserves the action $A_{\alpha'}[\gamma]=\int_\gamma \alpha'$ modulo $f$:
\begin{align}\label{modul1}
\frac{d}{ds}A_{\alpha'}[\gamma_s]\Big\vert_{s=0} &=\frac{d}{ds}\Big(\int_{\gamma_s}
\alpha+\beta\Big)\Big\vert_{s=0}\\
\nonumber&=\int_{\gamma_s} df\Big\vert_{s=0}=\int_a^b df(\gamma(t))=f(\gamma(b))-f(\gamma(a)),
\end{align}
where $\gamma: [a,b]\to M$ is an arbitrary smooth curve and the variation $\gamma_s$ is determined by the one-parameter group of
diffeomeomorphisms $g_s^\zeta$: $\gamma_s=g^\zeta_s(\gamma)$.

Now, let $\gamma: [a,b]\to M$ be a trajectory of \eqref{clb} and $\gamma_s=g^\zeta_s(\gamma)$. The relation $\dot\gamma(t)\in \mathcal L_{\gamma(t)}=\ker d\alpha'\vert_{\gamma(t)}$
and Cartan'a formula,
\begin{equation}\label{CartanF}
L_\zeta=i_\zeta\circ d+d\circ i_\zeta,
\end{equation}
imply
\begin{align}\label{modul2}
\frac{d}{ds}A_{\alpha'}[\gamma_s]\Big\vert_{s=0}&=\int_a^b
d(\alpha+\beta)(\zeta\vert_{\gamma(t)},\dot\gamma(t))dt+\int_a^b d\big((\alpha+\beta)(\zeta\vert_{\gamma(t)})\big)\\
\nonumber&=(\alpha+\beta)(\zeta\vert_{\gamma(b)})-(\alpha+\beta)(\zeta\vert_{\gamma(a)}),
\end{align}
and by comparing \eqref{modul1} and \eqref{modul2}, we obtain the identity
\[
(\alpha+\beta)(\zeta\vert_{\gamma(a)})-f(\gamma(a))=(\alpha+\beta)(\zeta\vert_{\gamma(b)})-f(\gamma(b)).
\]

Thus, the weak Noether symmetries induce conservation quantities described in the following statement (see \cite{Jo1}; for the Lagrangian setting and $M=\R\times TQ$, see \cite{SC1, Cr}).

\begin{thm}\label{prva}
Let $\zeta$ be a weak Noether symmetry of
equation \eqref{clb} that satisfies \eqref{wns1}. Then:

\medskip

{\rm (i)}  The function
\begin{equation}
\label{NetInt}
 J=i_\zeta(\alpha+\beta)-f
\end{equation}
  is a first integral of
\eqref{clb}.

\medskip

{\rm (ii)}  $J$ is preserved under the flow of
  $g_s^\zeta$ as well:
$L_\zeta(J)=0$.

\medskip

{\rm (iii)}  The commutator of vector fields $[Z,\zeta]$ is a
section of $\mathcal L$, i.e., $g_s^\zeta$ permutes the
trajectories of \eqref{clb} modulo reparametrization.
\end{thm}

It is natural to refer to \eqref{NetInt} as a \emph{Noether function} associated to the week Noether symmetry $\zeta$.

\subsection{Noether symmetries of time-dependent Hamiltonian equations}
The basic example is the extended phase space
endowed with the {Poincar\'e--Cartan} 1-form
\begin{equation}\label{cm}
(M,\alpha)=(\R\times T^*Q,pdq-Hdt).
\end{equation}
Namely, sections of $\mathcal L=\ker
d(pdq-Hdt)$ are of the form $Z_\mu=\mu(t,q,p) Z$, where $Z$ is the
vector field defining the Hamiltonian flow of $H$ in the extended phase space (e.g., see \cite{LM})
\begin{equation}\label{reb}
Z=\frac{\partial}{\partial
t}+\sum_i\big(\frac{\partial H}{\partial
p_i}\frac{\partial}{\partial q^i}- \frac{\partial H}{\partial
q_i}\frac{\partial}{\partial p_i}\big).
\end{equation}

The action functional \eqref{action} for $\alpha=pdq-Hdt$ implies
Poincar\'e's variant of the Hamiltonian principle of least action
\cite{Ar}, while Theorem \ref{prva} is a natural generalization of the classical Noether theorem (see subsection \ref{CNTP}).
Recently, a similar approach to the higher order Lagrangian problems is given in \cite{FiSp}.

\subsection{Noether integrals for perturbed systems}
Consider a perturbation \eqref{ZR} of equation \eqref{clb} by a vector field $P$, which is not a section of $\mathcal L=\ker d\alpha$.
Also, let $\mathcal M\subset M$ be an invariant submanifold of the system \eqref{ZR}, so we can consider the system restricted to $\mathcal M$.
The following observation, although quite elementary, is fundamental in our considerations.

\begin{thm}\label{druga}
Assume that $\zeta$ satisfies \eqref{wns1} restricted to $\mathcal M$. 
Then the derivative of the Noether function \eqref{NetInt} along the flow of \eqref{ZR} equals
\begin{equation}
\frac{d}{dt}J\big\vert_\mathcal M=d\alpha(P,\zeta)\vert_\mathcal M.
\end{equation}
\end{thm}

\begin{proof}
From the definition \eqref{wns1} and Cartan's formula \eqref{CartanF} we have
\begin{equation}\label{izvod2}
i_\zeta d\alpha=-d(i_\zeta(\alpha+\beta))+df=-dJ,
\end{equation}
implying the statement:
\begin{equation}\label{ZJ}
L_{Z+P}J=i_{Z+P}(-i_\zeta d\alpha)=d\alpha(Z,\zeta)+d\alpha(P,\zeta)=d\alpha(P,\zeta).
\end{equation}
\end{proof}

\begin{cor}
Assume that $P$ is $d\alpha$--orthogonal to the week Noether symmetry field $\zeta$ restricted to $\mathcal M$. Then $J\vert_\mathcal M$
is a first integral of system \eqref{ZR}.
\end{cor}

Next, from the proof of Theorem \ref{druga}, we see that we can relax the assumption that the vector field $\zeta$ is a week Noether symmetry.

\begin{thm}\label{treca}
Assume that $\zeta$ is a vector field and $\gamma$ is a 1-form satisfying
\begin{align}
\label{wns2}  & L_\zeta(\alpha+\beta)\vert_\mathcal M=df+\gamma\vert_\mathcal M,\\
\label{uslov} & d\alpha(P,\zeta)+i_{Z+P}\gamma=0\vert_\mathcal M.
\end{align}
Then the Noether function \eqref{NetInt} is preserved along the flow of \eqref{ZR}.
\end{thm}

\begin{proof}
Now, from the definition \eqref{wns2} and Cartan's formula we get
\begin{equation*}\label{izvod3}
i_\zeta d\alpha=-d(i_\zeta(\alpha+\beta))+df+\gamma=-dJ+\gamma.
\end{equation*}

Therefore
\begin{equation*}\label{ZJ2}
L_{Z+P}J=i_{Z+P}(-i_\zeta d\alpha)+i_{Z+P}\gamma=d\alpha(P,\zeta)+i_{Z+P}\gamma,
\end{equation*}
which proves the statement.
\end{proof}

Note that for $P=0$, $\mathcal M=M$, Theorem \ref{treca} implies the following variant of Theorem \ref{prva}:
if $\zeta$ satisfies \eqref{wns2}, where a 1-form $\gamma$ annihilates the line bundle $\mathcal L$,
then the Noether function \eqref{NetInt} is a first integral of
\eqref{clb}.

\section{Time-dependent systems with constraints and nonconservative forces}\label{S3}

\subsection{Equations} Consider a Lagrangian system $(Q,L,F)$, where $Q$ is a configuration
space, $L(t,q,\dot q)$ is a time-dependent Lagrangian, $L: \mathbb R\times
TQ\to \R$, and $F$ is a non-conservative force.
Assume that the motion of the system is subjected to $s$, in general time-dependent, independent ideal holonomic
constraints
\begin{equation}\label{constraints0}
f^l(t,q)=0, \qquad l=1,\dots,s,
\end{equation}
defining a $(n-s)$-dimensional time-dependent constraint submanifold $\Sigma_t\subset Q$.
Therefore, the velocities of the system satisfy the constraints
\begin{equation}
\label{constraints1}
a_0^l(t,q)+a_1^l(t,q)\dot q_1+\dots+a_n^l(t,q)\dot q_n=0, \qquad l=1,\dots,s.
\end{equation}
where
\[
a_0^l(t,q)=\frac{\partial f^l}{\partial t}(t,q), \quad a_i^l(t,q)=\frac{\partial f^l}{\partial q_i}(t,q), \quad i=1,\dots,n,
\]
and $q=(q_1,\dots,q_n)$ are local coordinates on $Q$.\footnote{As an example, we can take $Q$ to be the configuration space $\R^{3N}$ of $N$ free material points, see \cite{Ar, KK}. }

In addition, suppose $k-s$ additional
independent ideal (nonholonomic) constraints are given
\begin{equation}
\label{constraints}
a_0^l(t,q)+a_1^l(t,q)\dot q_1+\dots+a_n^l(t,q)\dot q_n=0, \qquad l=s+1,\dots,k.
\end{equation}
As a result, the velocities of the system
belong to a $(n-k)$-dimensional time-dependent affine distribution
\[
\mathcal D_t\subset T_{\Sigma_t}Q.
\]

Together with $\mathcal D_t$, we consider the associated
$(n-k)$--dimensional \emph{distribution of virtual displacements} $\mathcal D_t^0\subset T\Sigma_t\subset T_{\Sigma_t}Q$ defined by the homogeneous equations
\begin{equation}
\label{virtual}
\sum_{i=1}^n a_i^l(t,q)\xi_i=0, \qquad l=1,\dots,k.
\end{equation}

The motion of the system on the constrained space
\begin{equation}\label{CM}
\mathcal D=\{(t,\mathcal D_t)\,\vert\, t\in\R\}\subset \R\times TQ
\end{equation}
is described by the Euler--Lagrange--d'Alembert
equations
\begin{equation} \label{Lagrange1}
\sum_{i=1}^n \big(\frac{d}{dt}\frac{\partial L}{\partial \dot q_i}-\frac{\partial
L}{\partial q_i}-F_i\big)\xi_i=0,
\end{equation}
for all time-dependent vector fields $\xi=\sum_i\xi_i \partial/\partial q_i$, which satisfy the homogeneous constraints \eqref{virtual} (so called \emph{virtual displacements}).

Equivalently, equations \eqref{Lagrange1} can be rewritten in the form
\begin{equation} \label{Lagrange2}
\frac{d}{dt}\frac{\partial L}{\partial \dot q_i}-\frac{\partial
L}{\partial q_i}=F_i+R_i, \qquad i=1,\dots,n.
\end{equation}

Here, $F_i$ and
\begin{equation}\label{reactionF}
R_i=\sum_{l=1}^k \lambda_l(t,q,\dot q) a_i(t,q)^l, \qquad i=1,\dots,n,
\end{equation}
are the components of the nonconservative and the reaction force, respectively.
The Lagrange multipliers $\lambda_l=\lambda_l(t,q,\dot q)$ are determined from the condition that a motion $q(t)$ satisfy the constraints
\eqref{constraints0}, \eqref{constraints1}, \eqref{constraints}. We skip a discussion on the existence and the uniqueness of the Lagrange multipliers.

Let  $\mathbb FL_t: TQ\to T^*Q $ be the \emph{Legendre transformation}
\begin{equation}\label{legendre}
\mathbb FL_t(t,q,\xi)\cdot \eta=
\frac{d}{ds}\vert_{s=0}L(t,q,\xi+s\eta) \quad \Longleftrightarrow
\quad p_i=\frac{\partial L}{\partial \dot q_i}, \quad i=1,\dots,
n,
\end{equation}
where $\xi,\eta\in T_q Q$ and $(q_1,\dots,q_n, p_1,\dots,p_n)$ are
canonical coordinates of the cotangent bundle $T^*Q$.
In order to
have a Hamiltonian description of the dynamics we suppose that the
Legendre transformation \eqref{legendre} is a diffeomorphism.

Let
\[
\mathcal M=\{(t,\mathcal M_t)\,\vert\, t\in\R\}\subset \R\times T^*Q, \qquad \mathcal M_t\vert_q=\mathbb FL_t(\mathcal D_t\vert_q)\subset T^*_q Q
\]
be the constrained manifold \eqref{CM} within $\R\times T^*Q$. It is defined by the equations \eqref{constraints0} and
\begin{equation}
\label{constraints3}
a_0^l(t,q)+a_1^l(t,q)\frac{\partial H}{\partial p_1}+\dots+a_n^l(t,q)\frac{\partial H}{\partial p_n}=0, \qquad l=1,\dots,k.
\end{equation}

In the canonical coordinates
$(q,p)$ of the cotangent bundle $T^*Q$, the equations of motion
\eqref{Lagrange2} read:
\begin{equation} \label{1}
\dot q_i=\frac{\partial H}{\partial p_i},\qquad
\dot p_i=-\frac{\partial H}{\partial q_i}+F_i(t,q,p)+R_i(t,q,p), \qquad
i=1,\dots,n,
\end{equation}
where the Hamiltonian function $H(t,q,p)$ is the \emph{Legendre
transformation} of  $L$
\begin{equation} \label{haml}
 H(t,q,p)=\mathbb FL(t,q,\dot q)\cdot \dot q-L(t,q,\dot
q)\vert_{\dot q=\mathbb FL^{-1}(t,q,p)},
\end{equation}
and $F_i(t,q,p)=F_i(t,q,\dot q)\vert_{\dot q=\mathbb FL^{-1}(t,q,p)}$, $R_i(t,q,p)=R_i(t,q,\dot q)\vert_{\dot q=\mathbb FL^{-1}(t,q,p)}$.

In other words, on the constrained manifold $\mathcal M$ we have a system of the form \eqref{ZR}, where $Z$ is a section of the characteristic line bundle $\ker d(pdq-Hdt)$ given by \eqref{reb} and
the perturbation vector field is
\[
P=\sum_{i=1}^n \big(F_i(t,q,p)+R_i(t,q,p)\big)\frac{\partial}{\partial p_i}.
\]

\subsection{The reaction-annihilator distribution}
Let $(t,q,p)\in\mathcal M$.
Define distributions $\mathcal V$ and $\mathcal R$ of $T(\R\times T^*Q)$ at the points of $\mathcal M$
by
\begin{align*}
\mathcal V_{(t,q,p)} =&\big\{\zeta=\tau\frac{\partial
}{\partial t}+\sum_{i=1}^n{\xi_i\frac{\partial}{\partial
q_i}+\eta_i\frac{\partial}{\partial p_i}}\,\vert\,a_0^l(t,q)\tau+\sum_{i=1}^n a_i^l(t,q)\xi_i=0, \,l=1,\dots,k\big\},\\
\mathcal R_{(t,q,p)} =&\big\{\zeta=\tau\frac{\partial
}{\partial t}+\sum_{i=1}^n{\xi_i\frac{\partial}{\partial
q_i}+\eta_i\frac{\partial}{\partial p_i}}\,\vert\,\sum_{i=1}^n R_i(t,q,p)\big(\xi_i-\frac{\partial H}{\partial p_i}\tau\big)=0\big\}.
\end{align*}

We refer to
$\mathcal V$ as a \emph{admissible distribution} over $\mathcal M$, since the velocity of a curve $(t,q(t),p(t))$
belongs to $\mathcal V$ if and only if $q(t)$ satisfies the constraints \eqref{constraints0}, \eqref{constraints}.
The rank of $\mathcal V$ is $2n+1-k$. On the other hand,
the distribution $\mathcal R$ need not be of a constant rank: $\mathcal R_{(t,q,p)}$ is either a hyperplane or a whole tangent space $T_{(t,q,p)}\R\times T^*Q$, if reaction forces vanish at $(t,q,p)$.

Following \cite{FRS}, we call $\mathcal R$ a \emph{reaction-annihilator distribution} over $\mathcal M$ (see Section \ref{EXMP} below).

\begin{lem}\label{podskup} We have the inclusion $\mathcal V\subset \mathcal R$.
\end{lem}

\begin{proof}
Let $\zeta\in \mathcal V_{(t,q,p)}$. Then, from \eqref{reactionF} we have
\[
\sum_{i=1}^n R_i\xi_i=\sum_{i=1}^n\sum_{l=1}^k \lambda_l a_i^l \xi_i=-\tau\sum_{l=1}^k \lambda_l a_0^l.
\]
On the other hand, from \eqref{constraints3} we get
\begin{equation}\label{lepa}
\sum_{i=1}^n R_i\frac{\partial H}{\partial p_i}=\sum_{i=1}^n\sum_{l=1}^k \lambda_l a_i^l \frac{\partial H}{\partial p_i}=-\sum_{l=1}^k \lambda_l a_0^l,
\end{equation}
which implies $\zeta\in\mathcal R_{(t,q,p)}$.
\end{proof}

Let
\[
\zeta=\tau(t,q,p)\frac{\partial
}{\partial t}+\sum_i\xi_i(t,q,p)\frac{\partial}{\partial
q_i}+\eta_i(t,q,p)\frac{\partial}{\partial p_i}
\]
 be a {weak Noether symmetry} of the Hamiltonian system with the Hamiltoian \eqref{haml} restricted to $\mathcal M$:
\begin{equation}\label{wns}
L_{\zeta}(pdq-Hdt+\beta)=df\vert_\mathcal M,
\end{equation}
with respect to a closed 1-form $\beta$ and a smooth function $f$ in the extended phase space $\R\times T^*Q$.

From Theorem \ref{druga} we get.

\begin{thm}\label{cetvrta}
{\rm (i)} The derivative of
\[
J=i_\zeta(pdq-Hdt+\beta)-f=\sum_i p_i\xi_i-H\tau+\beta(\zeta)-f
\]
along the flow of \eqref{1} equals
\[
\frac{dJ}{dt}\big\vert_\mathcal M=\sum_{i=1}^n (F_i+R_i)(\xi_i-\dot q_i\tau)\vert_\mathcal M.
\]

{\rm (ii)} If $\zeta\vert_\mathcal M$ is a section of the admissible distribution $\mathcal V$, or more generally, a section of $\mathcal R$,
the derivative of the Noether function is given by
\[
\frac{dJ}{dt}\big\vert_\mathcal M=\sum_{i=1}^n F_i(\xi_i-\dot q_i\tau)\vert_\mathcal M.
\]
In particular, if $F\equiv 0$, $J$ is preserved along the flow of \eqref{1} if and only if $\zeta$ is a section of $\mathcal R$.
\end{thm}

\begin{proof}
(i) We have
\begin{align*}
d\alpha(P,\zeta)=&(dp\wedge dq-dH\wedge dt)(P,\zeta)=\sum_{i=1}^n dp_i(P)dq_i(\zeta)-dH(P)dt(\zeta)\\
=&\sum_{i=1}^n (F_i+R_i)\big(\xi_i-\tau \frac{\partial H}{\partial p_i}\big)=\sum_{i=1}^n (F_i+R_i)(\xi_i-\dot q_i\tau),
\end{align*}
where we used $dq_i(P)=dt(P)=0$.

\medskip

(ii) The proof follows directly from item (i) and Lemma \ref{podskup}.
\end{proof}

\subsection{Classical Noether theorem and prolongations of time-depended vector fields} \label{CNTP}
In the classical Noether theorem, without constraints and nonconservative forces, one considers the invariance of the action functional
\begin{equation}\label{action2}
A=\int_{t_0}^{t_1} L(t,q(t),\dot q(t))dt
\end{equation}
under the transformations induced by the $\mathbb R\times TQ$--prolongation (e.g., see \cite{SC1})
\begin{equation}\label{prol}
\hat\xi_{R\times TQ}=\tau\frac{\partial }{\partial
t}+\sum_i \xi_i\frac{\partial}{\partial q_i}+ \nu_i  \frac{\partial}{\partial \dot q_i}, \,\, \nu_i=\frac{\partial
\xi_i}{\partial t}-\dot q_i\frac{\partial \tau}{\partial t}+\sum_j\big(\frac{\partial \xi_i}{\partial q_j}\dot
q_j-\dot q_i\frac{\partial
\tau}{\partial q_j}\dot q_j\big)
\end{equation}
of a time-dependent vector field
\begin{equation}\label{symmetry}
\hat\xi=\tau\frac{\partial}{\partial t}+\xi=\tau(t,q)\frac{\partial}{\partial t}+\sum_i \xi_i(t,q)\frac{\partial}{\partial q_i}.
\end{equation}

The action \eqref{action2} is preserved if the Lagrangian satisfies the invariance condition:
\begin{align}
\label{uslovL}
\sum_i{\Big(\frac{\partial L}{\partial
q_i}\xi_i+\frac{\partial L}{\partial \dot
q_i}\nu_i\Big)}+\frac{\partial L}{\partial t}\tau+
L \big(\frac{\partial \tau}{\partial
t}+\sum_j\frac{\partial \tau}{\partial q_j}\dot q_j \big)=0,
\end{align}
and then
\begin{equation}\label{JL}
J(t,q,\dot q)=\frac{\partial L}{\partial \dot q}(\xi-\tau\dot
q)+L\tau=\sum_i{\frac{\partial L}{\partial \dot q_i}(\xi_i-\tau\dot
q_i)}+L\tau
\end{equation}
is a first integral of the Euler-Lagrange equations (e.g., see \cite{SC1}).

On the Hamiltonian side, firstly we need to take some natural prolongation of \eqref{symmetry} to $\R\times T^*Q$.
We define the prolongation
\begin{equation}\label{prolongation}
\zeta=\tau\frac{\partial}{\partial t}+\sum_i\xi_i\frac{\partial}{\partial q_i}+
\sum_{i}\big(H\frac{\partial\tau}{\partial q_i}-\sum_j\frac{\partial\xi_j}{\partial q_i}p_j\big)\frac{\partial}{\partial p_i}
\end{equation}
from the condition that the Lie derivative $L_\zeta(pdq-Hdt)$ is proportional to the one-form $dt$ (see Proposition 2.1 \cite{Jo1}).
Then $\zeta$ is a Noether symmetry of the Poincar\'e--Cartan 1-form
\begin{equation*}\label{ns}
L_\zeta(pdq-Hdt)=0
\end{equation*}
if and only if the following invariance condition is satisfied (see the proof of Proposition 2.2, \cite{Jo1}):
\begin{equation}\label{ns3}
L_\zeta H=p\frac{\partial\xi}{\partial t}- H\frac{\partial\tau}{\partial t}=\sum_{i=1}^n p_i\frac{\partial\xi_i}{\partial t}- H\frac{\partial\tau}{\partial t}.
\end{equation}

Moreover, under the Legendre transformation, the invariance conditions \eqref{uslovL} and \eqref{ns3} are equivalent, and the Noether integral
\begin{equation}\label{netIntPro}
J(t,q,p)=i_\zeta(pdq-Hdt)=\sum_{j=1}^n \xi_j(q,t)p_j-\tau(t,q)H(t,q,p)
\end{equation}
takes the usual form \eqref{JL} (Proposition 2.2 \cite{Jo1}).


Obviously, the first integrals of many classical problems are not of the form \eqref{JL}.
In order to extend the class of examples of Noether integrals,
the gauge terms and some modifications of the action \eqref{action2} are considered (see the references in \cite{SC1}).
Note that a function $f(t,q,p)$ in our definition of a week Noether symmetry plays the role of a gauge term
in the classical formulation.

The next natural step was the extending the class of symmetries.
The framework where the components of $\hat\xi$ additionally depend of velocities was introduced in \cite{D, DV},
and then a geometrical setting
for the equivalence of the first integrals and Noether symmetries of the Lagrangian given by vector fields on $\R\times TQ$ was formulated
in \cite{SC1, Cr}.


In our notation, if $\alpha$ is a contact form and $F$ is the integral of \eqref{clb},
we have the inverse Noether theorem directly, without using the gauge terms:
the contact Hamiltonian vector field of the function $F$ defines the Noether symmetry $\zeta$ of the equation \eqref{clb} with the Noether integral $F=i_\zeta\alpha$.

For mechanical problems, the {Poincar\'e--Cartan} 1-form is contact in a domain $U_H\subset \R\times T^*Q$ defined by the condition
\[
\rho=i_Z(pdq-Hdt)=p\frac{\partial H}{\partial p}-H\ne 0,
\]
where $Z$ is given by \eqref{reb}.
Thus, on $U_H$ we have a simple explicit expression
for the Noether symmetry (see Theorem 4.1 and examples given in  \cite{Jo1}). In this sense, the form $\beta$ should be taken such that
$pdq-Hdt+\beta$ is a contact form on $\R\times T^*Q$, i.e, $\rho+i_Z\beta\ne 0$. We left the terms $f$ and $\beta$ in the formulation of Theorem \ref{cetvrta}, although we do not use them in the examples given below.

\subsection{Noether theorem for prolonged vector fields} Now we return to the constrained system \eqref{1}.
Let
\[
\Sigma=\{(t,\Sigma_t)\vert\,t\in\R\}\subset \R\times Q
\]
be a $(n-s+1)$--dimensional submanifold of $\R\times Q$ defined by the holonomic constraints \eqref{constraints0}.
Consider a time-dependent vector field \eqref{symmetry} and its prolongation \eqref{prolongation}.
Since $\hat\xi$ does not depend on the momenta, the conditions on $\zeta\vert_\mathcal M$ to be a section of the admissible distribution $\mathcal V$, or a section of the reaction-annihilator distribution $\mathcal R$, reduce to the conditions on the vector field $\hat\xi\vert_\Sigma$.

It is clear that $\zeta\vert_\mathcal M$ belongs to $\mathcal V$ if and only if $\hat\xi\vert_\Sigma$ belongs to
a  $(n+1-k)$--dimensional distribution $\hat{\mathcal V}$ over $\Sigma$:
\[
\hat{\mathcal V}=\big\{\hat\xi\in T_{(t,q)}\Sigma\, \vert\,
a_0^l\tau+\sum_{i=1}^n a_i^l\xi_i=0, \, l=1,\dots,k,\, (t,q)\in\Sigma\big\}.
\]

Since the first $s$ equations in the definition of $\hat{\mathcal V}$ determine the tangent bundle $T\Sigma$,
we have
\[
\hat{\mathcal V}\subset T\Sigma\subset T_\Sigma(\R\times Q).
\]

On the other hand, by using the expressions for $R_i$ given by \eqref{reactionF} and relation \eqref{lepa}, $\zeta\vert_\mathcal M$ is a section of $\mathcal R$, if and only if
\begin{equation}\label{anh2}
\sum_{i=1}^n R_i\big(\xi_i-\tau\frac{\partial H}{\partial p_i}\big)=\sum_{l=1}^k \lambda_l(t,q,p)\big(a_0^l\tau +\sum_{i=1}^n a_i^l \xi_i\big)=0\vert_\mathcal M.
\end{equation}

That is, $\hat\xi\vert_{(t,x)}$, $(t,q)\in\Sigma$,
belongs to the subspace
$\hat{\mathcal R}_{t,q}$ of $T_{(t,q)}(\R\times Q)$
determined by equations \eqref{anh2} for all momenta $p$ in the space
\begin{equation}\label{lepa2}
\mathcal M_t\vert_q=\mathbb FL_t(\mathcal D_t\vert_q).
\end{equation}

Let $\hat{\mathcal R}$ be the collection of spaces $\hat{\mathcal R}_{t,q}$, $(t,q)\in\Sigma$.
We refer to $\hat{\mathcal V}$ and $\hat{\mathcal R}$ as a \emph{admissible} and a
\emph{reaction-annihilator} distribution over $\Sigma$, respectively. Obviously,
\[
\hat{\mathcal V}\subset \hat{\mathcal R}\subset T_\Sigma (\R\times Q).
\]

From Theorem \ref{cetvrta}, we have:

\begin{thm}\label{cetvrta*}
Let \eqref{symmetry} be a symmetry of the Hamiltonian: $H$ satisfies \eqref{ns3} on $\mathcal M$, where
 $\zeta$ is the prolongation of \eqref{symmetry} given by \eqref{prolongation}.\footnote{According to Proposition 2.2 \cite{Jo1}, this is
 equivalent to the assumption that the Lagrangian $L$ satisfies  \eqref{uslovL} at the constrained manifold $\mathcal D\subset \R\times TQ$. Then we have the following analogous statement in the Lagrangian setting: if $F\equiv 0$,
 the Noether function \eqref{JL} is conserved along the flow of the system \eqref{Lagrange2} if and only if
$\hat{\xi}\vert_{\Sigma}$ is a section of $\hat{\mathcal R}$.}
Then the derivative of the Noether function equals
\begin{equation}\label{momentum2}
\frac{d}{dt}J\vert_\mathcal M=\frac{d}{dt}\big(\sum_{j=1}^n\xi_j(q,t)p_j-\tau(t,q)H(t,q,p)\big)\vert_\mathcal M=\sum_{i=1}^n F_i(\xi_i-\dot q_i\tau)\vert_\mathcal M
\end{equation}
if and only if $\hat{\xi}\vert_{\Sigma}$ is a section of the reaction-annihilator distribution $\hat{\mathcal R}$.
Thus, for $F\equiv 0$, the Noether function \eqref{JL} is preserved if and only if $\hat\xi\vert_{\Sigma}$ is a section of
$\hat{\mathcal R}$.
In particular, $J$ is an integral   if $\hat\xi\vert_{\Sigma}$ is a section of the admissible distribution  $\hat{\mathcal V}$.
\end{thm}

Note that the left hand side of \eqref{momentum2} simplifies if $F$ is a gyroscopic force
\begin{equation}\label{GF}
\sum_{i=1}^n F_i\dot q_i=\sum_{i=1}^n F_i \frac{\partial H}{\partial p_i}=0.
\end{equation}

In the next section we shall consider two important special cases of Theorem \ref{cetvrta*}, which are variants of some
well known statements in nonholonomic mechanics.

\section{Examples}\label{EXMP}

\subsection{Nonholonomic Noether theorem}
If $\tau\equiv 0$ in the symmetry field \eqref{symmetry}, we have
a time-dependent vector field $\xi=\sum_j \xi_j(t,q)\partial/\partial q_j$ on $Q$ and the Noether function linear in momenta
\begin{equation}\label{linearInt}
J=\sum_{j=1}^n \xi_j(t,q)p_j.
\end{equation}

It is well known that the one-parameter group of diffeomorphisms $g_s^\xi$ of the vector field $\xi$ (for a fixed $t$)
have natural lifts to $TQ$ (used in considering the Noether symmetries in \cite{FRS}) and $T^*Q$.
A well known expression for the cotangent lift of $\xi$ is
\begin{equation}\label{linear}
\zeta=\sum_i\xi_i\frac{\partial}{\partial q_i}-
\sum_{i,j}\frac{\partial\xi_j}{\partial q_i}p_j\frac{\partial}{\partial p_i},
\end{equation}
which is the Hamiltonian vector field of the Noether function \eqref{linearInt} with respect to the canonical symplectic form $dp\wedge dq$ on $T^*Q$
(e.g., see Proposition 1.9, Ch IV, \cite{LM}).\footnote{Another way to see the prolongations of $\xi$ to $TQ$ and $T^*Q$ is simply the substitution of $\tau\equiv 0$ in \eqref{prol} and \eqref{prolongation}, respectively.}
Thus, the invariance condition \eqref{ns3} becomes
\begin{equation}\label{ns2}
L_\zeta H=p\frac{\partial\xi}{\partial t}=\sum_{i=1}^n p_i\frac{\partial\xi_i}{\partial t}.
\end{equation}

In particular, when $\xi$ does not depend on time, \eqref{ns2} reduces to the invariance of the Hamiltonian function with respect to the flow of \eqref{linear} on the cotangent bundle $T^*Q$,
or equivalently, the invariance of the Lagrangian $L$ under the action of the flow of $\xi$ that is extended to the tangent bundle $TQ$.
In that case, \eqref{linearInt} is the standard momentum map of the action of the symmetry field $\xi$ (see \cite{Ar, LM}).

For $\tau\equiv 0$, the equation \eqref{anh2} takes the form
\begin{equation}\label{anhilacija}
\sum_{i=1}^n R_i\xi_i=\sum_{i=1}^n\sum_{j=1}^k \lambda_l(t,q,p) a_i^l(t,q) \xi_i=0\vert_\mathcal M.
\end{equation}

As above, since $\xi$ does not depend on the momenta,
$\zeta\vert_\mathcal M$ belongs to $\mathcal R$ if and only if $\xi\vert_q$, $q\in {\Sigma_t}$, belongs to the subspace
$\mathcal R^0_{t,q}$ of $T_q Q$ defined by
equations \eqref{anhilacija} for all momenta $p$ in \eqref{lepa2}.
Let
\[
\mathcal R^0_t=\cup_{q\in\Sigma_t}\mathcal R^0_{t,q}\subset T_{\Sigma_t} Q.
\]

In the case of time-independent constrained Lagrangian systems, the distribution $\mathcal R^0_t$ is defined by Fass\`o, Ramos, and Sansonetto in \cite{FRS, FS}, where it is referred as a
\emph{reaction-annihilator distribution} (the sections of $\mathcal R^0_t$ annihilate the 1-form of reaction forces $\sum_i R_i dq_i$).
We keep the same notation.

The distribution of virtual displacements $\mathcal D_t^0$ defined by \eqref{virtual} is a subset of $\mathcal R_t^0$, and it takes the role
of the admissible distribution. Thus, we get

\begin{cor}\label{peta}
Let a time-dependent vector field $\xi$ on $Q$ be a symmetry of the Hamiltonian $H$: the invariance equation \eqref{ns2} is satisfied on $\mathcal M$, where
$\zeta$ is the prolongation \eqref{linear} of $\xi$. The momentum equation
\begin{equation}\label{momentum}
\frac{d}{dt}J\vert_\mathcal M=\frac{d}{dt}\big(\sum_{i=1}^n \xi_i(q,t)p_i\big)\vert_\mathcal M=\sum_{i=1}^n F_i\xi_i\vert_\mathcal M
\end{equation}
holds if and only if $\xi\vert_{\Sigma_t}$ is a section of the
reaction-annihilator distribution $\mathcal R_t^0$ for all $t$.
Thus, for $F\equiv 0$, the Noether function \eqref{linearInt} is a first integral
of the constrained system \eqref{1} if and only if $\xi\vert_{\Sigma_t}$ is a section of $\mathcal R_t^0$, $t\in\R$.
In particular,
it is an integral
if $\xi\vert_{\Sigma_t}$ is a section of the distribution of virtual displacements $\mathcal D_t^0$, $t\in\R$.
\end{cor}

This statement is referred as a \emph{nonholonomic Noether theorem}, usually formulated for the time-independent case and under the condition that $\xi$ is a symmetry of the nonconstrained system as well, induced by the action of an one-parameter subgroup of the Lie group $G$ of symmetries acting on the configuration space $Q$ (see \cite{Ar, BKMM}).

One of the first variants of the nonholonomic Noether theorem is given by Kozlov and Kolesnikov \cite{KK}. They considered a natural mechanical system of $N$ material points, $Q=\R^{3N}$, with symmetries that are infinitesimal isometries of $\R^{3N}$ with respect to the metric induced by the kinetic energy. For the rigid body examples, see, e.g., \cite{BM}. The Chaplygin rolling ball problem in $\R^n$ is an illustrative example of the reduction of symmetries of the preserved nonholonomic momentum map, see \cite{HN, Jo0}.

The situation where the function \eqref{linearInt} is not an integral of the nonconstrained system, while it is an
integral of the constrained one appears in the case of \emph{gauge symmetries} of nonholonomic systems studied in \cite{BS, BGMD, FRS, FGS, FS}. Corollary \ref{peta} for the Lagrangian systems $(Q,L)$ with $F\equiv 0$ and time-independent nonholonomic constraints \eqref{constraints} can be found in \cite{FRS, FGS} (the homogeneous case) and \cite{FS} (the affine case).

Thus, we can consider Theorems \ref{cetvrta}, \ref{cetvrta*} as a generalization of gauge symmetries to the time-dependent symmetries, as well as to the symmetries which are related to integrals that are not linear in momenta.

\begin{exm}\label{Example1} As an illustration,
consider the motion of a material point with the position vector $\mathbf r=(x,y,z)$, mass $m$ and electric charge $q$ in $\R^3$ under the influence of a
homogeneous gravitational field $\mathbf F=mg\mathbf k$ in the direction of a unit vector $\mathbf k=(k_x,k_y,k_z)$ and a magnetic field $\mathbf B=B\mathbf e_z=(0,0,B)$.
The Newtonian equations of motion are
\[
\dot{\mathbf p}=mg\mathbf k+\epsilon \dot{\mathbf r} \times \mathbf e_z,
\]
where $\epsilon=qB$, $\mathbf p=m\dot{\mathbf r}=(p_x,p_y,p_z)$. For the Hamiltonian we take the kinetic energy $H_0=\frac{1}{2m}(\mathbf p,\mathbf p)$ and treat the gravitational field as a nonconservative force.

Now, assume the motion  is subjected to a time-dependent nonholonomic constraint:
\[
a(t)y\dot x-\dot z+b(t)=0.
\]
The symmetry vector field $\xi=\mathbf e_y$ belongs to the distribution of virtual displacements. Thus, from the momentum equation
\[
\dot{p}_{y}=F_y+B_y=mg k_y-\epsilon \dot x,
\]
we get that
$p_{y}-mg k_yt+\epsilon x$ is an integral of the system.

 Consider the vector field $\xi=f \mathbf e_x+a(t)yf \mathbf e_z$, which also belongs to the distribution of virtual displacements. Let $\zeta$ be a prolongation of $\xi$:
\[
\zeta=f \mathbf e_x+a(t)yf \mathbf e_z-\big(\frac{\partial f}{\partial y}p_x+a(t)f p_z+a(t)y\frac{\partial f}{\partial y}p_z \big)\mathbf e_{p_y},
\]
for $f=f(t,y)$. The invariance condition \eqref{ns2} takes the form
\[
-\frac{1}{2m}\big( \frac{\partial f}{\partial y} p_xp_y+af p_yp_z+ay\frac{\partial f}{\partial y}p_yp_z    \big)=p_x \frac{\partial f}{\partial t}+p_z y
\big(   \frac{\partial a}{\partial t}f+a\frac{\partial f}{\partial t}  \big),
\]
and
for $b\equiv 0$,
under the condition
\[
p_z=a(t)y p_x,
\]
 we have a solution
\[
f(t,y)=1/\sqrt{1+a(t)^2y^2}
\]
($\xi$ is a time-dependent gauge symmetry).

Thus, the momentum equation reads
\[
\frac{d}{dt}\big(\frac{p_x}{\sqrt{1+a(t)^2y^2}}+\frac{a(t)y p_z}{\sqrt{1+a(t)^2y^2}}\big)\big\vert_{p_z=a(t)y p_x}=
mg\frac{k_x+a(t)y k_z}{\sqrt{1+a(t)^2 y^2}}+\epsilon\frac{\dot y}{\sqrt{1+a(t)^2 y^2}}
\]
and for $k_x=k_z=\epsilon=0$, we have the additional gauge integral
\[
\Big(\frac{p_x}{\sqrt{1+a(t)^2 y^2}}+\frac{a(t)y p_z}{\sqrt{1+a(t)^2 y^2}}\Big)\big\vert_{p_z=a(t)y p_x}\,.
\]
\end{exm}

\subsection{Conservation of energy}
As the next important case,  we suppose $\tau\equiv 1$ and that
there is no influence of nonconservative forces ($F\equiv 0$).
Consider a time-dependent vector field
\begin{equation}\label{moving0}
\hat\xi=\frac{\partial}{\partial t}+\xi=\frac{\partial}{\partial t}+\sum_i \xi_i(t,q)\frac{\partial}{\partial q_i}
\end{equation}
and its prolongation
 (see \eqref{prolongation}):
\begin{equation}\label{moving}
\zeta=\frac{\partial}{\partial t}+\sum_i\xi_i\frac{\partial}{\partial q_i}-
\sum_{i,j}\frac{\partial\xi_j}{\partial q^i}p_j\frac{\partial}{\partial p_i}.
\end{equation}

As in the previous subsection, it is a Noether symmetry if the invariance condition \eqref{ns2} holds (but now with $\zeta$ given by \eqref{moving} instead by \eqref{linear}) and
the associated Noether function takes the form
\begin{equation}\label{netIntPro2}
J(t,q,p)=i_\zeta(pdq-Hdt)=\sum_{j=1}^n \xi_j(q,t)p_j-H(t,q,p).
\end{equation}

It is well known that if the Hamiltonian $H$ does not depend
on time, then $\zeta=\partial/\partial t$ ($\xi\equiv 0$) is a Noether symmetry and the Hamiltonian multiplied by $-1$ ($J=-H$) is the corresponding
Noether integral of the non-constrained system.

Note that the affine distribution $\mathcal D_t\subset T_{\Sigma_t} Q$ can be seen as a sum $\mathcal D^0_t+{\xi}^0_t$, where
${\xi}_t^0=\sum_j \xi_j^0(t,q)\partial/\partial q_j$ is a  section of the affine distribution $\mathcal D_t$ over $\Sigma_t$, $t\in\R$.

\begin{cor}\label{sesta}
{\rm (i)} Let \eqref{moving0} be a symmetry vector field of the Hamiltonian $H$: the equation \eqref{ns2} is satisfied on $\mathcal M$,
where $\zeta$ is the prolongation \eqref{moving} of $\hat\xi$.
Then the constrained system has a Noether integral
\eqref{netIntPro2} if and only if $\xi\vert_{\Sigma_t}$ is a section
of the affine distribution $\xi^0_t+\mathcal R_t^0$ for all $t$.
In particular, if $\xi\vert_{\Sigma_t}$ is a section of the affine bundle $\mathcal D_t$, $t\in\R$, then $J$ is an integral of the system.

\medskip

{\rm (ii)}  Assume that the Hamiltonian $H$ does not depend of time. It is preserved if and only if the vector field $\xi^0_t$ is a section
of the reaction-annihilator distribution $\mathcal R_t^0$, $t\in\R$.
In particular, if the holonomic constraints \eqref{constraints0} do not depend on time and the nonholonomic constraints \eqref{constraints} are homogeneous (and possibly time-dependent), the Hamiltonian
is a first integral of the system.
\end{cor}

\begin{proof}
(i) According to Theorem \ref{cetvrta*}, the constrained system has a Noether integral
\eqref{netIntPro2} if and only if $\hat\xi\vert_\Sigma$ is a section of $\hat{\mathcal R}$, i.e., $\hat\xi\vert_{\Sigma}$ is a solution of equation \eqref{anh2}. Since $\xi^0_t$ satisfies the constraints \eqref{constraints1}, \eqref{constraints},
\[
a_0^l(t,q)+a_1^l(t,q)\xi^0_1+\dots+a_n^l(t,q)\xi^0_n=0, \qquad l=1,\dots,k,
\]
it follows that the vector field $\xi\vert_{\Sigma_t}-\xi^0_t$ is a solution of \eqref{anhilacija} (i.e., $\xi\vert_{\Sigma_t}-\xi^0_t$ is a section of $\mathcal R^0_t$, $t\in\R$) if and only if $\hat\xi\vert_{\Sigma}$ is a solution of equation \eqref{anh2}.

Further, for the last statement, we note that $\hat\xi\vert_\Sigma$ is a section of the admissible distribution $\hat{\mathcal V}$ if and only if  $\xi\vert_{\Sigma_t}$ is a section of $\mathcal D_t$, $t\in\R$.

\medskip

(ii) The statement follows from item (i) by taking $\xi\equiv 0$ in \eqref{moving0}: the zero section is a section of
$\xi^0_t+\mathcal R_t^0$ if and only if $\xi^0_t$ is a section
of $\mathcal R_t^0$.
In particular, if the holonomic constraints \eqref{constraints0} do not depend on time and the nonholonomic constraints \eqref{constraints} are homogeneous, then $\mathcal D_t\equiv \mathcal D_t^0$ and we can take $\xi^0_t \equiv 0$.
\end{proof}

A variant of Corollary \ref{sesta} is given in \cite{FS, FS2}, where the Noether function $J=\sum_j\xi_j p_j-H$ is referred as a \emph{moving energy} integral.
For the conservation of energy in systems with affine constraints see also \cite{BMB}.

\begin{rem}{\rm The Hamiltonian is also conserved if in item (i) of Corollary \ref{sesta} we assume the action of
a gyroscopic force $F$ (see \eqref{GF}).
}\end{rem}

\begin{exm}
{\rm Let
\[
H=H_0+V=\frac{1}{2m}(\mathbf p,\mathbf p)-mg(\mathbf k,\mathbf r)
\]
be the total energy of a material point considered
in Example \ref{Example1}. According to item (i) of Corollary \ref{sesta} and the above remark, it is conserved for the nonholonomic problem with a time-dependent homogeneous constraint
\[
a(t)y\dot x-\dot z=0.
\]
}\end{exm}

\begin{rem}{\rm
It would be interesting to have a geometrical setting for the Noether theorem
in quasi-coordinates given in \cite{D, M, SS}. The case, where the symmetries are induced by the symmetry vector field $\xi$ on the configuration space $Q$ is treated, e.g., in \cite{BMZ}. Our goal is also the analysis of a symmetry and integrability of the multidimensional rolling spheres problems introduced in \cite{Jo2}.
}\end{rem}

\subsection*{Acknowledgments} The author is very grateful to the referees for pointing misprints and various suggestions
that improved the exposition of the results.
The research was supported by the Serbian Ministry of
Science Project 174020, Geometry and Topology of Manifolds,
Classical Mechanics and Integrable Dynamical Systems.


\begin{thebibliography}{84}

\bibitem{Ar}
V. I. Arnold, V. V. Kozlov, A. Neishtadt, \emph{Mathematical Aspects of Classical and Celestial
Mechanics}. Dynamical Systems, III. Third Edition. Encyclopaedia Math. Sci. 3. Springer,
Berlin, 2006.

\bibitem{BS} P. Balseiro, N. Sansonetto,
\emph{A Geometric Characterization of Certain First Integrals for Nonholonomic Systems with Symmetries},
SIGMA  {\bf 12} (2016), 018, 14 pages, arXiv:1510.08314.


\bibitem{BGMD}
L. Bates, H. Graumann, C. MacDonnell, \emph{Examples of gauge conservation laws in nonholonomic systems}. Rep. Math. Phys. {\bf 37} (1996), 295--308.

\bibitem{BKMM}
A.M. Bloch, P.S. Krishnaprasad, J.E. Marsden, R.M. Murray, \emph{Nonholonomic mechanical systems with symmetry}. Arch. Rational Mech. Anal. {\bf 136} (1996), 21--99.


\bibitem{BMZ}
A. M. Bloch, J. E. Marsden, D. V. Zenkov, \emph{Quasivelocities and symmetries in nonholonomic systems}, Dynamical Systems, {\bf 24} (2009), 187--222.

\bibitem{BM} A. V. Borisov, I. S. Mamaev,
\emph{Symmetries and Reduction in Nonholonomic Mechanics},  Regular and Chaotic Dynamics, {\bf 20} (2015) 553--604

\bibitem{BMB} A. V. Borisov, I. S. Mamaev, I. A. Bizyaev,
\emph{The Jacobi Integral in Nonholonomic Mechanics},  Regular and Chaotic Dynamics, {\bf 20} (2015) 383--400.


\bibitem{SC1} F. Cantrjin, W. Sarlet, \emph{Generalizations of Noether's
theorem in classical mechanics}, SIAM Review {\bf 23} (1981)
467--494.

\bibitem{C}
F. Cantrjn, M. de Leon, M. de Diego, J. Marrero, \emph{Reduction of nonholonomic mechanical
systems with symmetries}, Rep. Math. Phys, {\bf 42} (1998), 25--45.

\bibitem{Cr} M. Crampin,
\emph{Constants of the motion in Lagrangian mechanics},
Int. J. Theor. Phys. {\bf 16} (1977), No. 10, 741--754.

\bibitem{CM} M. Crampin, T. Mestdag
\emph{The Cartan form for constrained Lagrangian systems and the
nonholonomic Noether theorem}, Int. J. Geom. Methods. Mod. Phys. {\bf 8} (2011), 897--923,  	arXiv:1101.3153.


\bibitem{D} Dj. S. Djuki\'c, \emph{Conservation laws in classical mechanics for quasi-coordinates}.
Arch. Ration. Mech. Anal. {\bf 56} (1974), 79--98.

\bibitem{DV} Dj. S. Djuki\'c, B. D. Vujanovi\'c, \emph{Noether's Theory in the Classical Nonconservative Mechanics}, Acta Mech.
{\bf 23}(1975) 17--27.


\bibitem{FRS}
F. Fass\`o, A. Ramos, N. Sansonetto, \emph{The reaction-annihilator distribution and the
nonholonomic Noether theorem for lifted actions}. Regul. Chaotic Dyn. {\bf 12} (2007), 449--458.

\bibitem{FGS} F. Fass\`o, A. Giacobbe, N. Sansonetto, \emph{Gauge conservation laws and the momentum equation in nonholonomic mechanics},
Reports in Mathematical Physics, {\bf 62} (2008) 345--367.


\bibitem{FS}
F. Fass\`o, N. Sansonetto,
\emph{Conservation of  energy and momenta in nonholonomic systems with affine constraints},
Regular and Chaotic Dynamics,   {\bf 20} (2015), 449--462,  	arXiv:1505.01172.

\bibitem{FS2}
F. Fass\`o, N. Sansonetto, \emph{Conservation of `moving' energy in nonholonomic systems with affine constraints
and integrability of spheres on rotating surfaces}, J. Nonlinear Sci. {\bf 26} (2016), 519-–544, arXiv:1503.06661.


\bibitem{FiSp}  E. Fiorani,  A. Spiro, \emph{Lie algebras of conservation laws of variational ordinary differential equations}, J. Geom. Phys.
{\bf 88} (2015), 56--75,  	arXiv:1411.6097

\bibitem{Gi}
G. Giachetta, \emph{First integrals of non-holonomic systems and their generators}, J. Phys. A:
Math. Gen. {\bf 33} (2000) 5369–-5389.

\bibitem{HN} S. Hochgerner, L. C. Garcia-Naranjo, \emph{G-Chaplygin systems with internal symmetries, truncation, and an
(almost) symplectic view of Chaplygin’s ball}. J. Geom. Mech. {\bf 1} (2009) 35–-53. arXiv:0810.5454.


\bibitem{Jo0} B. Jovanovi\' c, \emph{Hamiltonization and Integrability of the
Chaplygin Sphere in $\R^n$}, {J. Nonlinear. Sci.} {\bf 20} (2010)
569--593, arXiv:0902.4397.

\bibitem{Jo1} B. Jovanovi\' c, \emph{Noether symmetries and integrability in Hamiltonian time-dependent mechanics}, Theoretical and Applied Mechanics,  \textbf{43} (2016) 255--273.

\bibitem{Jo2}  B. Jovanovi\' c, \emph{Invariant measures of modified LR and L+R systems}, Regular and Chaotic Dynamics,   {\bf 20} (2015), 542--552,
 	arXiv:1508.04913.

\bibitem{KK} V.V. Kozlov, N.N. Kolesnikov, \emph{On theorems of dynamics}. (Russian) J. Appl. Math. Mech. {\bf 42}
(1978), no. 1, 28-–33.

\bibitem{KS} Y. Kosmann-Schwarzbach, \emph{The Noether Theorems, Invariance and Conservation Laws in the Twentieth Century}, Springer, New York, 2011.

\bibitem{LM} P. Libermann, C. Marle, \emph{Symplectic Geometry, Analytical Mechanics}, Riedel, Dordrecht, 1987.

\bibitem{Noether} E. Noether, \emph{Invariante Variationsprobleme}, Nachrichten von der K\"oniglich Gesellschaft der Wissenschaften zu G\"ottingen, Mathematisch-physikalische Klasse (1918), 235--257.

\bibitem{Marle}
C.-M. Marle, \emph{On symmetries and constants of motion in Hamiltonian systems with non-
holonomic constraints}. In \emph{Classical and Quantum Integrability} (Warsaw, 2001), 223--242,
Banach Center Publ. {\bf 59}, Polish Acad. Sci. Warsaw, 2003.

\bibitem{M} Dj. Mu\v sicki, \emph{Noether's theorem for nonconservative systems in quasicoordinates}, Theoretical and Applied Mechanics {\bf 43} (2016), 1--17.

\bibitem {SS} S. Simi\'c,
\emph{On Noetherian approach to integrable cases of the motion of heavy top},
Bull. Cl. Sci. Math. Nat. Sci. Math. No. 25 (2000), 133--156.

\end{thebibliography}
\end{document}